\documentclass{amsart}
\usepackage{epsfig}
\usepackage{latexsym}
\usepackage{amsmath}
\usepackage{graphicx}
\usepackage{amsmath}
\usepackage{amssymb}
\newtheorem{thm}{Theorem}

\newtheorem{corollary}[thm]{Corollary}
\theoremstyle{definition}

\newcommand{\ignore}[1]{}

\newcommand{\cE}{\mathcal{E}}

\newcommand{\PP}{\mathbb{P}}
\newcommand{\EE}{\mathbb{E}}
\newcommand{\R}{\mathcal{R}}
\newcommand{\Q}{\mathcal{Q}}

\newcommand{\D}{\mathcal{D}}

\begin{document}

\author{Wim Hordijk and Mike Steel} 

\title{Autocatalytic Sets in Polymer Networks with Variable Catalysis Distributions}

\address{WH: Konrad Lorenz Institute for Evolution and Cognition Research,
Klosterneuburg, Austria. MS: Biomathematics Research Centre, University of Canterbury, Christchurch,
New Zealand}
\maketitle

\begin{abstract}
All living systems -- from the origin of life to modern cells -- rely on a set of biochemical reactions that are simultaneously self-sustaining and autocatalytic. This notion of an autocatalytic set has been formalized graph-theoretically (as `RAF'), leading to mathematical results and polynomial-time algorithms that have been applied to simulated and real chemical reaction systems. In this paper, we investigate the emergence of autocatalytic sets in polymer models when the catalysis rate of each molecule type is drawn from some probability distribution. We show that although the average catalysis rate $f$ for RAFs to arise depends on this distribution, a universal linear upper and lower bound for $f$ (with increasing system size) still applies. However, the probability of the appearance (and size) of autocatalytic sets can vary widely, depending on the particular catalysis distribution. We use simulations to explore  how tight the mathematical bounds are, and the reasons for the observed variations. We also investigate the impact of inhibition (where molecules can also inhibit reactions) on the emergence of autocatalytic sets, deriving new mathematical and algorithmic results. \end{abstract}

\section{Autocatalytic sets and RAF theory}

The notion of self-sustaining autocatalytic sets was introduced and studied by Kauffman \cite{Kauffman:71,Kauffman:86,Kauffman:93} in the context of the  origin of life \cite{Hordijk:10}. The concept was later formalized mathematically and further developed as RAF theory \cite{Steel:00,Hordijk:04,Hordijk:13a}. In this section, we briefly review the basics of RAF  theory and its main results.

First, we define a {\it chemical reaction system} (CRS) as a tuple $\Q=(X,\R,C)$ consisting of a set $X$ of molecule types, a set $\R$ of chemical reactions, and a catalysis set $C$ indicating which molecule types catalyze which reactions. Formally, each reaction $r \in \R$ can be viewed as a pair of subsets of $X$ (the `reactants' and `products' of the reaction $r$) and $C$ is a subset of $X \times \R$.
We also assume throughout that $X$ contains a distinguished subset, a  {\em food set}  $F \subseteq X$, which is a subset of molecule types that are assumed to be directly available from the environment (i.e., they do not necessarily have to be produced by any of the reactions in the system).

A CRS can be graphically represented by a {\it bipartite directed hypergraph}. In such a graph representation there are two types of nodes: (1) molecule types (black dots) and (2) chemical reactions (white boxes). There are also two types of (directed) edges: (1) reaction edges (solid black) and (2) catalysis edges (dashed grey). The end points (incident nodes) of these edges are always from different node sets (i.e., an edge is always directed from a molecule node to a reaction node, or vice versa). A simple example is shown in Fig. \ref{fig:RAFexample}.

The notion of catalysis plays a central role here. A {\it catalyst} is a molecule type that significantly speeds up the rate at which a chemical reaction happens but without being ``used up'' in that reaction. Catalysis is ubiquitous in life \cite{Santen:06}. Almost all organic reactions are catalyzed, and catalysts are essential in determining and regulating the functionality of the chemical networks that support life.

Given a CRS $\Q=(X,\R, C)$ and food set $F \subseteq X$, an autocatalytic set, or {\it RAF set}, is now defined as a nonempty subset $\R' \subseteq \R$ of reactions which is:
\begin{enumerate}
  \item {\it Reflexively Autocatalytic} (RA): each reaction $r \in \R'$ is catalyzed by at least one molecule type that is either a product of another reaction in $\R'$ or is an element $F$.
  \item {\it Food-generated} (F): each molecule type that is a reactant of any reaction $r \in \R'$ can be generated from the food set $F$ by using a series of reactions only from $\R'$ itself.
\end{enumerate}

A simple example of a RAF set consisting of just two reactions is presented in Fig. \ref{fig:RAFexample}, though a RAF set can of course be of any size. A  more formal mathematical definition of RAF sets was provided in \cite{Hordijk:04,Hordijk:11}, including an efficient (polynomial-time) algorithm for finding such sets in arbitrary CRSs (or determining that there is none). This RAF algorithm is based on a simple pruning rule that iteratively removes reaction nodes from the input CRS that (by definition) cannot be part of a RAF set. If the algorithm returns a non-empty subgraph, this subgraph then represents the {\it maximal} RAF set (i.e., the union of all RAF sets) contained within the input CRS. Otherwise, there is no RAF set in the given CRS.

\begin{figure}[htb]
\centering
\includegraphics[scale=0.9]{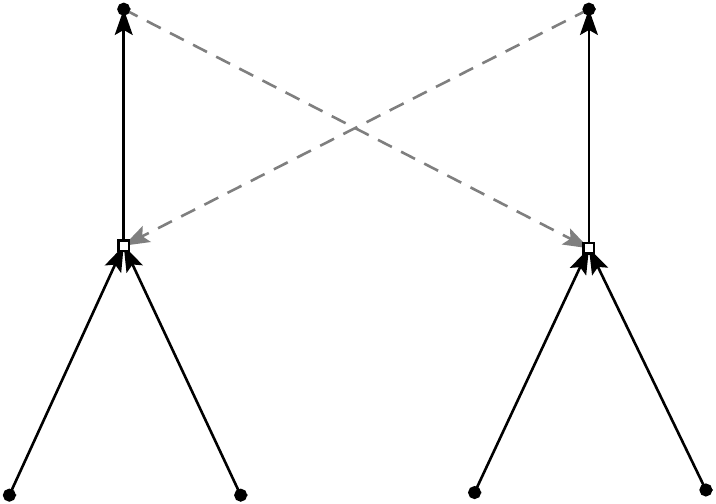}
\caption{An example of a simple CRS with just two reactions, represented as a bipartite directed hypergraph. Black dots represent molecule types and white boxes represent chemical reactions. Solid (black) arrows indicate reactants going into and products coming out of a reaction. Dashed (grey) arrows indicate catalysis. With a food set consisting of the four molecule types at the bottom, this CRS forms an autocatalytic (RAF) set where the two products (at the top) mutually catalyze each other's formation.}
\label{fig:RAFexample}
\end{figure}

Note that the example RAF set in Fig. \ref{fig:RAFexample} requires at least one ``spontaneous'' (i.e., uncatalyzed) reaction event before the system forms a RAF (i.e. at least one of the two reactions needs to happen spontaneously before the first catalyst will be present when starting with only molecules from the food set). This is, of course, always possible, but at a much lower rate than when these same reactions are catalyzed. As a consequence, there may be a (stochastic) ``waiting time'' before this RAF set is realized in a dynamical sense. However, once a RAF is realized, it can, in principle, grow in concentration at an exponential rate, due to its (collectively) autocatalytic nature.

As argued elsewhere, this requirement for (rare) spontaneous reactions is actually a useful property for the potential evolvability of autocatalytic sets \cite{Hordijk:14b}. Furthermore, it is the main reason why catalysts are treated separately in RAF theory, instead of including them in a reaction as both a reactant and a product (which is often done in alternative chemical reaction network representations). Such an alternative representation requires the catalyst to always be present (as a reactant), preventing reactions from occasionally happening spontaneously.

RAF theory has been applied extensively to simple polymer-based models of chemical reaction networks (see next section), showing that autocatalytic sets are highly likely to exist in such models, and also for chemically realistic levels of catalysis. More specifically, it was shown both experimentally and theoretically that only a {\it linear} growth rate in the level of catalysis with increasing maximum polymer length is required to get RAF sets with high probability \cite{Hordijk:04,Mossel:05}, where the level of catalysis is expressed in terms of the average number of reactions catalyzed per molecule type. This is in stark contrast to the required {\it exponential} growth rate in the original argument of Kauffman \cite{Kauffman:86,Kauffman:93}. Also, as opposed to the original assumption that autocatalytic sets appear as ``giant connected components'' making up almost the entire reaction network \cite{Kauffman:86,Kauffman:93}, RAF sets in the simple polymer model tend to be much smaller, consisting of only a fraction of the full reaction network \cite{Hordijk:04}. In fact, the collection of RAF sets under the partial order of set inclusion is generally very large, with a unique maximal element (the maximal RAF), all the way down to the minimal RAF subsets (called {\em irreducible} RAFs) \cite{Hordijk:12a}.

Furthermore, these main results have been shown to hold under a wide variety of chemically more realistic model assumptions, such as template-based catalysis (as with base-pair formation in RNA) \cite{Hordijk:11}, allowing only the longest polymers to be catalysts \cite{Hordijk:14}, allowing partitioned chemical reaction networks (as in nucleic acids vs. proteins) \cite{Smith:14}, or using power-law distributed catalysis \cite{Hordijk:14a}.

However, autocatalytic sets are not just a theoretical construct, as they have also been created and studied in real chemical networks under controlled laboratory conditions \cite{Sievers:94,Ashkenasy:04,Lincoln:09,Vaidya:12}. In fact, the formal RAF framework was used to analyze one of these real autocatalytic networks in detail  \cite{Hordijk:13}, one consisting of 16 catalytic RNA molecules, or ribozymes \cite{Vaidya:12}. Moreover, it was shown, using the RAF algorithm, that the metabolic network of {\it E. coli} forms a large autocatalytic set of close to 1800 reactions \cite{Sousa:15}. As far as we know, this is the first formal proof that living organisms (or at least essential parts thereof) can be realised as a RAF.

Finally, we have shown that ``higher levels'' of autocatalytic sets can emerge \cite{Hordijk:12a,Hordijk:15a}. For example, a boundary (such as a lipid layer) can be considered as an additional catalyst: it increases the rate at which reactions happen inside it by keeping the relevant molecules in close proximity rather than having them diffuse away, though the boundary itself is not used up in those reactions. This way, an ``autocatalytic set of autocatalytic sets'' emerges, which can form a simple (proto)cell-like structure \cite{Hordijk:15a}. This line of reasoning can, of course, be extended to the next emergent level of autocatalytic sets forming multicellular organisms, and so on, all the way up to the species level.

\section{Polymer models} \label{sec:BinPolModel}
Polymers, such as RNA sequences and the amino acid sequences that comprise enzymes and proteins, play a central role in biochemistry and are likely to have also played an important role in early life. In previous work, we have considered a simple polymer-based chemical reaction system known as the {\it binary polymer model}, originally introduced by Kauffman \cite{Kauffman:86,Kauffman:93}. In this model, polymers are represented by bit strings that can be concatenated into longer ones or split up into shorter ones. Catalysis is assigned at random. Instances of this model then form a full CRS as follows:
\begin{itemize}
  \item The molecule set $X$ consists of all bit strings up to (and including) a maximum length $n$. In mathematical notation: $X = \{0,1\}^{\leq n}$.
  \item The food set $F$ consists of all bit strings up to (and including) length $t$, where $t \ll n$. In our experiments, we use $t=2$, i.e., $F = \{0,1,00,01,10,11\}$.
  \item The reaction set $\R$ consists of two types of reactions:
  \begin{enumerate}
    \item {\it Ligation}: gluing two bit strings together into a longer one, e.g., $00 + 111 \rightarrow 00111$.
    \item {\it Cleavage}: splitting a bit string into two smaller ones, e.g., $010101 \rightarrow 01 + 0101$.
  \end{enumerate}
All possible ligation and cleavage reactions between bit strings are included in $\R$, as long as none of their reactants or products violates the maximum molecule length $n$. Note that each ligation reaction has a corresponding cleavage reaction. In fact, in most of our experiments, these corresponding ``forward'' and ``backward'' reactions are counted as one bi-directional reaction.
  \item The catalysis set $C$, which consists of pairs $(x,r)$ of a molecule types and reaction,  is  assigned according to some random process.  In the original (uniform) model, each molecule type $x$ has the same probability of catalyzing each reaction $r$. However, biochemical networks (e.g. the {\em E. coli} study from \cite{Sousa:15}) suggest that catalysis is far from uniform and that the distribution of catalysis events across molecule types has a longer tail that is closer to a power law than a binomial distribution. Thus a focus of this paper is to explore the different distributions of catalysis. We will present some new mathematical results that hold for a general distribution of catalysis rates, but we also focus on four particular ways of assigning catalysis:
  \begin{enumerate}
    \item {\it Uniform}: For each  molecule type $x \in X$ and reaction $r \in \R$, the pair $(x,r)$ is independently included in $C$ with probability $p$. This is the method used in the original model \cite{Kauffman:86,Kauffman:93}.
    \item {\it Power law}: For each molecule type $x \in X$,  first a random number $k_x\geq 1$ is drawn from a Zipf distribution with parameter $a$:
  \[ \PP(k_x = k) = \frac{k^{-a}}{\zeta(a)}, \]
where $\zeta$ is the Riemann Zeta function and $k_x$ ranges over the positive integers.  Next, molecule type $x$ is assigned as a catalyst to $k'_x-1$ distinct but randomly chosen reactions from the full set $\R$ (which are drawn independently for each molecule type), where $k'_x=\min \{k_x, |\R|\}$.  Note that we subtract one from the value of $k'_x$, so there are many molecule types that do not catalyze any reaction. This method was first investigated in \cite{Hordijk:14a}.
  \item {\it All-or-nothing}: Each molecule type $x \in X$ independently either catalyzes {\it all} reactions in $\R$ with probability $p$, or $x$ catalyzes no reaction with probability $1-p$. This is an ``extreme'' model that we use for comparison below.
  \item{\it Sparse}: Each molecule type $x \in X$ independently  either (i) catalyzes each reaction $r \in \R$ independently with probability $\frac{1}{n}$ (this occurs with probability $p$), or 
(ii) $x$ does not catalyze any  reaction (this occurs with probability $1-p$). This is a less extreme version of the all-or-nothing model.
  \end{enumerate}
In each assignment method, all reaction are considered to be bi-directional. In other words, if a molecule type is assigned as a catalyst to a reaction, it catalyzes both the forward and backward directions.
\end{itemize}
For a given value of the maximum molecule length $n$, the set of molecule types $X$ and the set of reactions $\R$ are fixed. In other words, the nodes (both the molecule nodes and the reaction nodes) and the reaction edges (represented graphically by solid black lines) are fixed in the reaction graph. However, each instance of the model (even for the same values of $n$ and $p$) will have a different catalysis assignment $C$ (i.e. a  different configuration of catalysis edges, represented graphically by dashed grey lines). Fig. \ref{fig:RAF_1} shows an example of a (maximal) RAF set that was found by our RAF algorithm in an instance of the binary polymer model with parameters $n=5$, $t=2$, and $p=0.0045$ (using the uniform catalysis assignment method).

\begin{figure}[htb]
\centering
\includegraphics[scale=0.7]{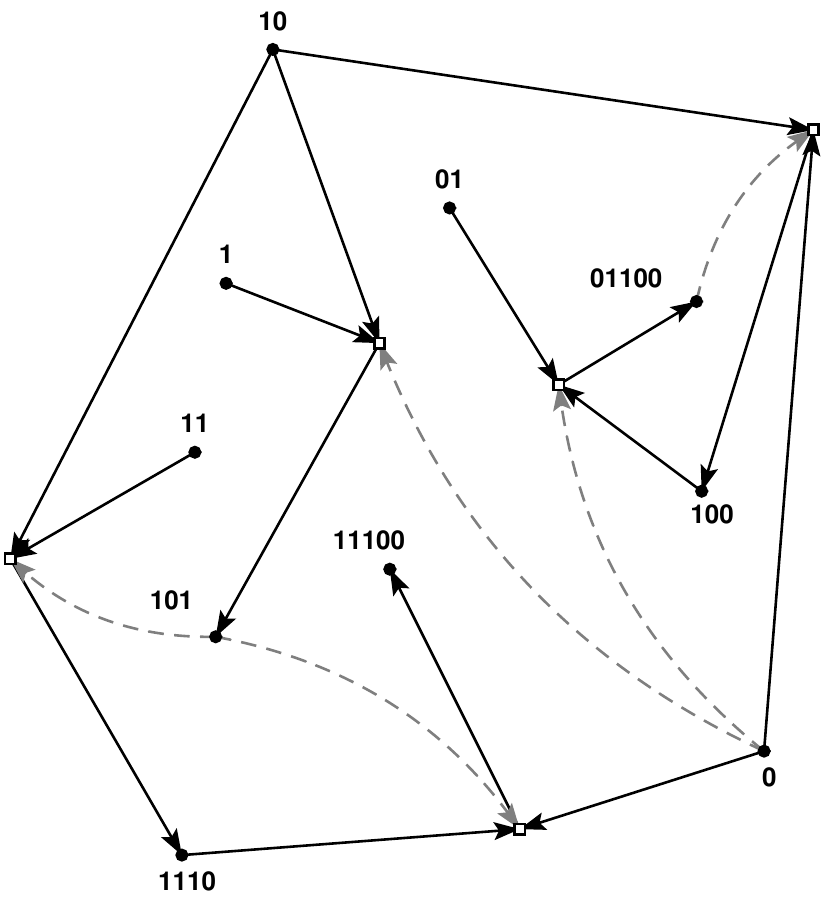}
\caption{An example of a (maximal) RAF set as found by our algorithm in an instance of the binary polymer model with $n=5$, $t=2$, and $p=0.0045$ (uniform catalysis). Note that the CRS that this RAF lies within has $|X|=62$ molecule types and $|\R|=196$ reactions.}
\label{fig:RAF_1}
\end{figure}

As already mentioned above, the main results of RAF theory hold up under various alternative versions of the basic model that introduce more chemical realism. In fact, in some cases it is even possible to predict mathematically the required level of catalysis in these more elaborate versions from that in the standard model \cite{Hordijk:12,Hordijk:14}. So, the binary polymer model, even in its simplest version, already seems to hold much predictive and explanatory power. Moreover, the restriction to polymers that are binary is in no way serious, since, as we will see, most of the mathematical results translate directly to polymers over an arbitrary alphabet. For the simulations, however, it is convenient to stay with the binary polymer model.

In the next section, we derive new mathematical results that describe the probability of RAFs arising as the average catalysis rate varies in a class of polymer models that includes the four catalysis assignment methods described above. In Sections~\ref{levels1} and \ref{levels2}, we then use simulations of the binary polymer model to see how close these mathematical bounds come to the estimated probabilities of RAFs, and provide an explanation for the differences. 

\subsection{Inhibition ($u$-RAFs)} In this paper, we will also consider a version of the binary polymer model where molecules can {\it inhibit} reactions (i.e., prevent them from happening) in addition to catalyzing them. More specifically, similar to the catalysis set $C$, we define an inhibition set $I$ of molecule--reaction pairs $(x,r)$, where $(x,r) \in I$ means that molecule type $x \in X$ {\em inhibits} the reaction $r \in \R$ (this could be represented by a dashed red edge, for example, in the graph representation of the CRS). These inhibition pairs are also chosen at random for each instance of the model, using the uniform method described above for assigning catalysis, and independent of the catalysis assignments. 

 In this setting, a CRS that allows inhibition is a tuple $(X,\R,C,I)$ (with corresponding food set $F$). The definition of a RAF $\R'$ can then be extended to require, in addition, that no reaction in $\R'$ is inhibited by any molecule type that is produced by a reaction in $\R'$ or present in the food set $F$.
 
 We refer to such an ``uninhibited'' RAF as a $u$-RAF.  We present a mathematical results on
$u$-RAFs in Section~\ref{impact} and simulation results in  Section \ref{sec:Inhibition}.




\newpage
\section{Mathematical results on RAF (and $u$-RAF) sets in polymer models}
\label{mathsy}

In this section, we allow a little more generality in considering polymers over an arbitrary alphabet of $\kappa$ letters. To emphasize the dependence of the CRS on $n$, we will write $\Q_n=(X_n, \R_n, C)$, where $X_n$ is the set of all polymers up to length $n$ and $\R_n$ is the associated set of all cleavage-ligation reactions.  A fundamental relationship between the number of molecule types and reactions in $\Q_n$ (Eqn. (4) from \cite{Mossel:05}) is:
\begin{equation}
\label{asyeq}
1-O\left(\frac{1}{n}\right) \leq \frac{|\R_n|}{n|X_n|} \leq 1.
\end{equation}
In particular, $|\R_n| \sim n|X_n|$. 

We will also generalize catalysis beyond the four models described above to a setting where each molecule type $x$ catalyzes $Y_x$ reactions chosen uniformly at random from $\R_n$, where the random variable $Y_x$ is independently and identically distributed (i.i.d.) according to some distribution $\D$ (which may depend on $n$); we will call such models {\em distribution-based}. For example, in the uniform model, $Y_x$ has the binomial distribution ${\rm Bin}(|\R_n|, p)$, while in the sparse model, $Y_x$ has the mixture distribution described by letting 
$Y_x$ follow a binomial distribution ${\rm Bin}(|\R_n|, \frac{1}{n})$ with probability $p$, and setting $Y_x = 0$ with probability $1-p$.

For the uniform model, any collection of events of the form $$(x_1,r_1) \in C, (x_2, r_2) \in C,  \ldots, (x_m, r_m) \in C$$ comprise stochastically independent events.  However, for other distributions (including the other three mentioned above) these events can fail to be independent if the molecule types $x_i$ are not distinct.  
For example,  the two events $(x,r) \in C$ and $(x, r') \in C$ will be positively correlated in general (since although the events are conditionally independent given $Y_x$ these two conditional probabilities
both increase with $Y_x$). This subtlety means that certain mathematical arguments  applicable to the uniform model need to be carefully adjusted to extend to distribution-based models (the models we study here do not fall under the umbrella of the `extended binary polymer model' of \cite{Hordijk:13a}, which required the stronger independence condition to hold).

We will let $f$ denote the (average) {\em catalysis rate} in a CRS:
$$f = \sum_{r \in \R} \PP((x,r) \in C),$$
for any $x \in X$ (by the symmetry in the model, different choices of $x$ lead to the same value). Notice also that we can write $f = |\R_n| \PP((x,r) \in C)$ for any choice of $x \in X$ and $r \in R$. For example, in the uniform model $f = |\R_n| p$ while in the sparse model, $f=p|\R_n|/n \sim p|X_n|$.

A main goal in RAF theory has been to estimate the probability $P_n$ that an instance of the  uniform random polymer model with parameter $n$ contains a RAF. Bounds on $P_n$ (for the uniform model) were initiated in \cite{Kauffman:86}, refined in \cite{Steel:00}, and extended further by simulations in \cite{Hordijk:04}, which led to a more detailed mathematical analysis in  \cite{Mossel:05}. 

A main finding from this last paper was that RAFs arise in $\Q_n$ precisely when the average number of reactions each molecule type catalyzes ($f$) grows (at least) linearly  with $n$. This follows directly from the following pair of inequalities, both of which hold for all values of $n$. Recall that $F$ (the food source) consists of all polymers up to a fixed length $t$. For any distribution-based polymer model, we let $\tilde{P_n}$ be the probability that an instance of the model has a RAF in which every molecule type in $X$ is generated from $F$ using just the products of ligation reactions. By definition, $P_n \geq \tilde{P_n}$. For the uniform model we have:

\begin{itemize}
\item[(i)] If  $f \geq \lambda  n$, then 
\begin{equation}
\label{basiceq}
P_n \geq \tilde{P_n} \geq 1- \frac{\kappa(\kappa e^{-\lambda})^t}{1-\kappa e^{-\lambda}},
\end{equation} 
which converges quickly to 1 as $\lambda$ grows, and 
\item[(ii)]
if  $f \leq \lambda n$, then 
\begin{equation}
\label{basiceq2}
P_n = O(\lambda) \rightarrow 0 \mbox{ as } \lambda \rightarrow 0.
\end{equation}
\end{itemize} 

This linear rate of growth in $f$ required for RAFs to arise in the uniform polymer model contrasts sharply with the provably exponential rate required for a RAF in which every catalyst must already be available the first time the reaction proceeds (i.e., when no spontaneous reactions are allowed) \cite{Mossel:05}.

What can we say about the level of catalysis required for the emergence of RAFs in other (non-uniform) distribution-based polymer models?

For the all-or-nothing model, there is a very precise expression for $P_n$. If we write $f= \lambda n$ then:
\begin{equation}
\label{pneq}
P_n = 1- \exp(-\lambda) + o(1),
\end{equation}
where the $o(1)$ term decays to zero at an exponential rate with increasing $n$. The proof of Eqn.~(\ref{pneq}) is essentially trivial, since in the all-or-nothing model we have $f = p|\R_n|$ and so if we set $f=\lambda n$, then $p = \lambda n / |\R_n|$ and a RAF exists whenever at least one molecule type catalyzes one (and therefore every) reaction. This occurs with probability
$$1-(1-p)^{|X_n|} = 1 - \left(1- \frac{\lambda n}{ |\R_n|}\right)^{|X_n|} \sim 1-\exp(-\lambda),$$
by the relation $n|X_n| \sim |\R_n|$ from Eqn. (\ref{asyeq}). 

What is curious here is that in the all-or-nothing model, the same (linear) dependence of $P_n$ on $n$ occurs, even though the proof from \cite{Mossel:05} of the inequality (\ref{basiceq}) for the uniform model breaks down in one crucial step for the all-or-nothing model. However, the resulting RAFs in the all-or-nothing model  are very different from those in the uniform model. For a start, in the all-or-nothing model, RAFs arise at the same catalysis level required for all of $\R_n$ to be a RAF. By contrast, in the uniform model, the catalysis rate required for $\R_n$ to be a RAF grows {\em quadratically} with $n$. More precisely, in the uniform model, if we let $f= \lambda n^2$ then:
\begin{equation}
\label{allR}
\lim_{n \rightarrow \infty} \PP(\R_n  \mbox{  is a RAF}) =
 \begin{cases}
1, &  \mbox{ if } \lambda > \ln \kappa;\\
0, & \mbox{ if } \lambda < \ln \kappa.
\end{cases}
\end{equation}
A proof of this sharp transition result, which improves on a result from \cite{Steel:00}, is given in the Appendix.

A further difference between the all-or-nothing and the uniform model is that, for the latter model and at the linear catalysis rate at which RAFs appear, `small' RAF are highly unlikely to 
arise (more precisely, with probability converging to 1 as $n$ grows, any RAF must be larger than a certain function that grows exponentially with $n$)  \cite{Steel:13}. 
By contrast, in the all-or-nothing model, as soon as there is a RAF there is one of order $n$:  namely, any catalyst and any sequence of $O(n)$ reactions that  generates it from $F$.

The analysis of the sparse model is more interesting, and will also see shortly how the sparse model will also have small RAFs (of size $O(n^2)$) when RAFs appear, again in contrast to the uniform model.

We now state our first main mathematical result. 

\begin{thm}
\label{mainthm}
For any distribution-based polymer model, a necessary and sufficient condition for RAFs to arise is that the average catalysis rate $f$ grows (at least) linearly with $n$. More precisely:
\begin{itemize}
\item[(i)] If $f \geq \lambda n$ (and  $\lambda > 2 \ln \kappa$),  then 
\begin{equation}
\label{beq2}
P_n \geq \tilde{P_n} \geq \left(1- \frac{\kappa(\kappa e^{-\lambda/2})^t}{1-\kappa e^{-\lambda/2}}\right)(1-e^{-\lambda/2}),
\end{equation} 
which converges to 1 at exponential rate as $\lambda$ grows.
\item[(ii)]  If $f \leq \lambda n$ then for a constant $M$ (independent of $n$), we have:
\begin{equation}
\label{beq}
P_n \leq M \lambda \rightarrow 0 \mbox{ as }  \lambda \rightarrow 0.
\end{equation}
\end{itemize}
\end{thm}

\begin{proof} We first show that a key identity  (derived for the uniform model in  \cite{Mossel:05}) holds for  all distribution-based polymer models. Let $q_*$ be the probability that any given reaction $r \in \R_n$ is {\em not} catalyzed by any molecule type in $X_n$. By the assumptions of the model, this probability is the same for each choice of $r$ and is given by:
\begin{equation}
\label{qeq}
q_*= \prod_{x \in X_n} \left(1-\frac{f}{|\R_n|}\right) \sim \exp(-f/n),
\end{equation}
where the asymptotic equivalence is from Eqn.~(\ref{asyeq}). The proof of Eqn. (\ref{qeq}) exploits the independence of the events $\cE_{x,r}:=$ ``$x \mbox { catalyzes } r$'' when $r$ is fixed and $x$ varies to give:
$$q_* = \PP(\bigwedge_{x \in X_n} \overline{\cE_{x,r}}) = \prod_{x\in X_n} \PP(\overline{\cE_{x,r}})= \prod_{x\in X_n} (1- \PP(\cE_{x,r})),$$
and Eqn.(\ref{qeq}) now follows since $p = \PP(\cE_{x,r})$ and $f= |\R_n| p$.

To establish the lower bound inequality on $\tilde{P_n}$ (Part (i)), we first observe that for any non-negative random variable $W$ with finite expected value $\mu = \EE[W]$, the following inequality holds:
\begin{equation}
\label{eh}
\PP(W <  \sqrt{\mu}) \geq 1- \sqrt{\mu}.
\end{equation}
This inequality follows from the standard inequality $\mu \geq \PP(W \geq w)w$ applied to $w=\sqrt{\mu}$ to deduce that $\PP(W \geq \sqrt{\mu}) \leq \sqrt{\mu}$, from which the result follows immediately. Let $N$ be the number of reactions in $\R_n$ that are not catalyzed by any molecule type in $X_n$, and let $W = N/|\R_n|$. We have $\EE[W] = q_*$. Let $\tilde{P_n}(W)$ be the probability that the polymer model has a RAF that generates every molecule type in $X_n$ conditional on the random variable $W$ (= $N/|\R_n|$). Given certain  reactions $r_1, r_2,\ldots, r_s$, let $\overline{\cE_i}$ be the event that reaction $r_i$ is not catalyzed by any molecule type in $X_n$. Then, by exchangeability, we have the following expression for the joint conditional probability:
$$\PP( \overline{\cE_1}  \wedge \overline{\cE_2} \wedge \cdots \wedge \overline{\cE_s}|N ) = \frac{N}{|\R_n|} \times \frac{N}{|\R_n|} \times \cdots \times \frac{N-(s-1)}{|\R_n|} \leq \left(\frac{N}{|\R_n|}\right)^s = W^s.$$
We can now apply the argument in the proof of Theorem 4.1(ii) of \cite{Mossel:05} to deduce that:
\begin{equation}
\label{helpsus}
\tilde{P_n}(W) \geq 1- \frac{\kappa(\kappa W)^t}{1-\kappa W}.
\end{equation}
By the total law of expectation, for any $w$, we have:
$$\tilde{P_n} = \EE[\tilde{P_n}(W)]  = \EE[\tilde{P_n}(W)| W < w] \PP(W  <  w]  + \EE[\tilde{P_n}(W)| W \geq w] \PP(W \geq w].$$
Since the second term on the right is non-negative, we have:
$$\tilde{P_n} \geq \EE[\tilde{P_n}(W)| W < w] \PP(W  <  w].$$
If we now set $w= \sqrt{\EE[W]} = \sqrt{q_*}$, then Inequality (\ref{eh}) and Eqn. (\ref{helpsus}) gives:
\begin{equation}
\label{n1}
\tilde{P_n} \geq  1- \frac{\kappa(\kappa \sqrt{q_*})^t}{1-\kappa \sqrt{q_*}} (1-\sqrt{q_*}).
\end{equation}
Finally, from Eqn. (\ref{qeq}) and $f \geq \lambda n$, we obtain:
\begin{equation}
\label{n2}q_* \leq e^{-\lambda},
\end{equation} 
Combining Inequalities (\ref{n1}) and (\ref{n2}) and noting that the expression in (\ref{n1}) is decreasing with $q_*$, we obtain the inequality required for Part (i).

\bigskip

Turning to Part (ii), we now establish Inequality~ (\ref{beq}). 
From  Eqn. (\ref{qeq}), we have:
$$q_{*} =  \left(1- \frac{f}{|\R_n|}\right)^{|X_n|}  \geq 1- \frac{ f |X_n|}{|\R_n|}.$$
Therefore,  if $f \leq \lambda n$, 
\begin{equation}
\label{helps1} 
1-q_{*} \leq \lambda +o(1),
\end{equation}
by Eqn.~(\ref{asyeq}).
The number $M$ of reactions that take two food molecules and produce another molecule is, at most, $2|F|^2= 2 x_t^2$, where $x_t = \kappa + \kappa^2+ \cdots +\kappa^t$. For there to be a RAF, at least one of these $M$ forward reactions must be catalyzed. By Boole's inequality, the probability of this is bounded above by $M(1-q_*)$, and so, from Eqn. (\ref{helps1}), $P_n$ is less than or equal to  $M(\lambda+o(1))$, which converges to zero as $\lambda  \rightarrow 0$. 
\end{proof}

Fig.~\ref{fig:compare} shows the behavior of $P_n$ as a function of $f$ for the binary polymer model, across our four exemplar distributions (for $n=10$ and $n=16$). The plots were obtained by simulations and the use of the RAF algorithm (except for the all-or-nothing model, where the exact formula was used). 
\begin{figure}[htb]
\centering
\includegraphics[width=9cm]{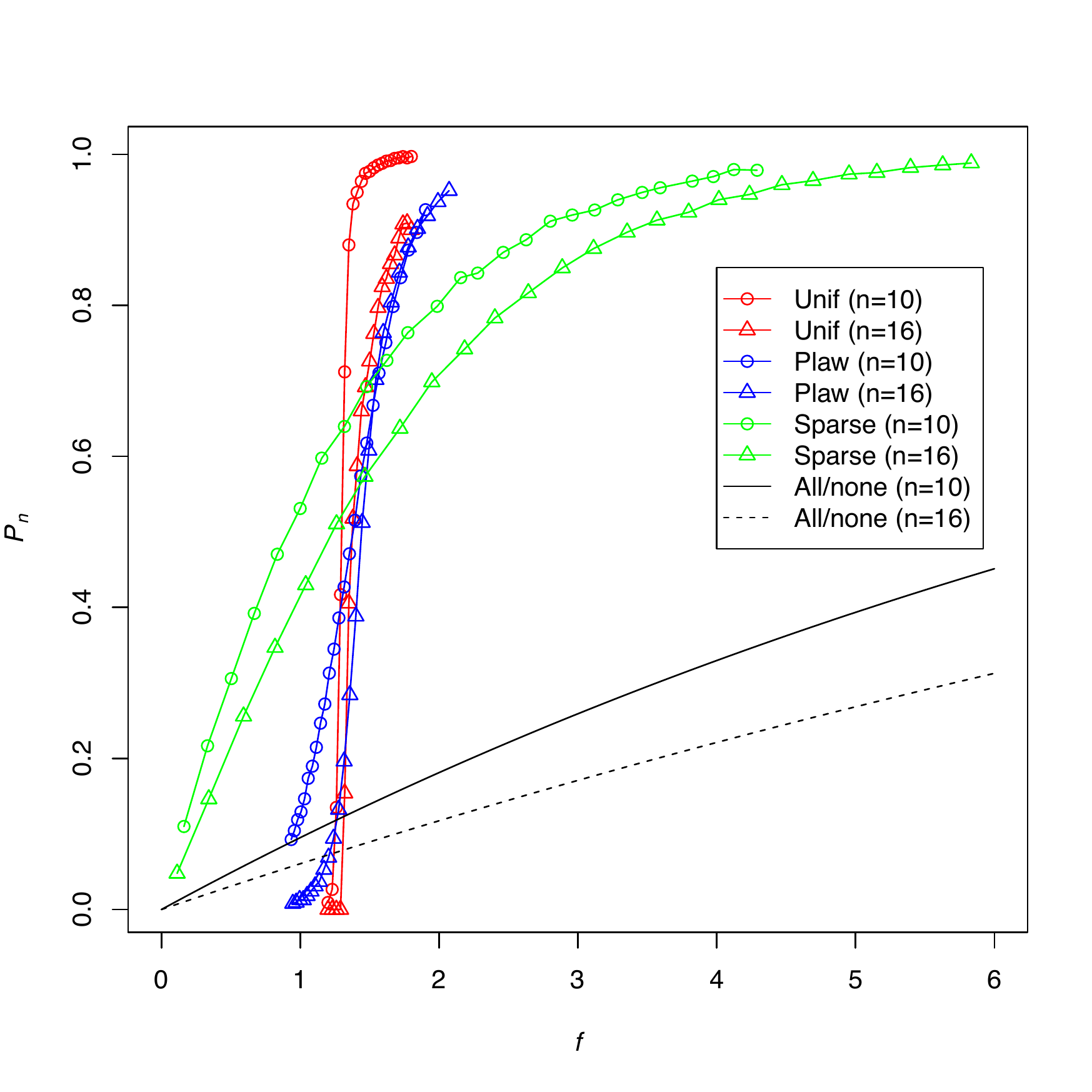}
\caption{Comparison across the four models of $P_n$ (probability of a RAF) as a function of $f$ (average catalysis rate per molecule) on the binary polymer model for $n=10$ and $n=16$.}
\label{fig:compare}
\end{figure}
Two features are apparent. Firstly, the uniform model and the power law model show a much sharper transition from $P_n=0$ to $P_n=1$ than the other two models. Secondly, the effect of increasing $n$ (from 10 to 16) has a larger effect on the curves for the sparse and all-or-nothing model than it does for the uniform and power law models.
Theorem~\ref{mainthm} also leads to the following interesting consequence for the structure of small RAFs in the sparse model, which is  different to that for the uniform model.

\begin{corollary}
Consider the sparse polymer model with $f= \lambda n$, where $\lambda$ is sufficiently large such that (from Theorem~\ref{mainthm})  $\tilde{P_n} \geq 1-\epsilon$. Then with probability at least $1 -\epsilon - o(1)$,  any instance of the model will contain a RAF with less than $2\lambda n^2$ reactions (where $o(1)$ refers to a term that converges to zero exponentially with  increasing $n$).
\end{corollary}

\begin{proof}  In the sparse polymer model, the number of catalysts (i.e. molecule types that catalyze one or more reactions) has mean $\lambda n$, and so (by the Chernoff bound)  the probability that there are more than $2\lambda n$ catalysts differs from 1 by a term that converges to zero exponentially with increasing $n$.  If $\tilde{P_n} \geq 1-\epsilon$ then with probability at least $1-\epsilon$, there is a RAF that constructs all molecule types in $X_n$ from $F$ using just the products of ligation reactions. Combining these two observations it follows that with probability at least $1-\epsilon - o(1)$, there is a RAF whose set $S$ of catalysts has size at most $2\lambda n$ and which constructs each molecule type $x \in X_n$ from $F$ via some sequence 
$s(x) = r_1(x), r_2(x), \ldots, r_m(x)$ of ligation reactions. Moreover, we must have $m < n$ (since each ligation reaction increases the size of a molecule by at least one),  so 
the set $(r_i(x): x\in S)$ is a RAF of size less than  $2\lambda n \times n = 2\lambda n^2$.
\end{proof}

\subsection{The impact of inhibition}
\label{impact}

Consider now the uniform model in which, in addition to  catalysis, each molecule type independently inhibits each reaction with a constant  probability. Let $h$ be the average number of reactions each molecule type inhibits, and, as before, let  $f$ be the average number of reactions each molecule type catalyzes ($f=p|\R_n|)$. Recall that a $u$-RAF is a RAF $\R'$ for which no reaction in $\R'$ is inhibited
by any molecule type in the set $S(\R')$ consisting of the products of reactions in $\R'$ together with the molecule types in the food set $F$. 

Clearly if $h>0$ then the probability of a $u$-RAF is bounded above by a  constant $\delta(h)< 1$ (even as $f \rightarrow \infty$) since there is a positive probability that all of the reactions that just involve
reactants from the food set $F$ are inhibited. Moreover $\delta(h) \rightarrow 0+$ as $h$ grows, even if $f$ were to grow exponentially (or faster) with $h$.
Despite this `trumps--all' feature of inhibition, one can still provide an inequality that relates the expected number of $u$-RAFs in a uniform polymer model with rates $(2f, h)$ to the expected number of RAFs in the uniform model with catalysis rate $f$.  

Let us assume that the catalysis rate  $f$ grows linearly with $n$ (this is a minor restriction since this rate suffices for RAFs to arise) and allow $h$ to also grow linearly with $n$ but with a constant factor that decays exponentially with the corresponding constant factor for $f$. Under this restriction, Theorem~\ref{mainthm2} (below) shows that we expect (at least) as many $u$-RAFs as there would be RAFs if the catalysis rate had simply been halved. 

Given a uniform polymer model $\Q_n$ (over any finite alphabet of $\kappa$ states)  let $\mu(f,h)$ be the expected number of $u$-RAFs as a function of  $f$ and $h$ (the catalysis and inhibition rates, respectively) and let
$\mu(f)$ be the expected number of RAFs as a function of $f$.   

\begin{thm}
\label{mainthm2}
Suppose that $f \leq \lambda n \leq \frac{1}{2}|\R_n|$, and  $h \leq \beta f$ where $\beta \leq n$ satisfies
\begin{equation}
\label{constraint}
\beta \leq \frac{1}{\lambda}\ln\left[\left(1- \frac{\beta\lambda n}{|\R_n|}\right)  \left(1+ e^{-\lambda}\left(1-\frac{\lambda n}{|\R_n|}\right)\right)   \right] \simeq \frac{\ln(1+e^{-\lambda})}{\lambda}.
\end{equation}
Then
$\mu(2f, h) \geq \mu(f).$
\end{thm}

\begin{proof}
For a subset $\R'$ of $\R$, let  $\cE_{\R'}$ be the  event that  $\R'$ is a $u$-RAF.
Then 
\begin{equation}
\label{mueq}
\mu(f,h) = \sum_{\R' \in \bf{R}_{\rm{FG}}} \PP(\cE_{\R'}),
\end{equation}
where $\bf{R}_{\rm{FG}}$ be the collection of subsets $\R'$ of $\R_n$ that are $F$-generated.
In the uniform model, we can write:
\begin{equation}
\label{fir}
\PP(\cE_{\R'}) = \theta_{(f,h,s)}^{|\R'|},
\end{equation}
where $$\theta_{(f,h,s)}  = \left[1 - \left(1-\frac{f}{|\R_n|}\right)^{s}\right] \left( 1- \frac{h}{|\R_n|}\right)^{s},$$
and where $s=|S(\R')|$ (the number of molecule types produced by reactions in $\R'$ or present in the food set). 
If we now let $\nu = \frac{f}{|\R_n|}$, then 
\begin{equation}
\label{x1} 
 \theta_{(f,0,s)}  = [1-(1-\nu)^s].
\end{equation}
Moreover,
\begin{equation}
\label{x2} 
\theta_{(2f,h,s)}  \geq  [1-(1-2\nu)^s][1-\beta \nu]^s \geq [1-(1-\nu)^{2s}][1-\beta\nu]^s,
\end{equation}
where the first inequality is because $h \leq \beta f$.  The second inequality in (\ref{x2})  follows from the following generic inequality: 
$$(1-2y)^k \leq (1-y)^{2k},$$
for all integers $k\geq 1$ and $0 \leq y \leq \frac{1}{2}$ (this can be readily verified by induction on $k$)
and noting that $0 \leq \nu = f/|R_n| \leq \frac{1}{2}$ by the assumption on $f$ in the statement of the theorem. 

Combining  Eqns. (\ref{x1}) and (\ref{x2}) and applying the identity $(1-y^2)/(1-y) = (1+y)$ we find that the
ratio of $\theta_{(2f,h,s)}$ to $\theta_{(f,0,s)}$ satisfies 
\begin{equation}
\label{firstfrac}
\frac{\theta_{(2f,h,s)}}{\theta_{(f,0,s)}} \geq  [1+(1-\nu)^s][1-\beta \nu]^s.
\end{equation}
Since $f \leq \lambda n$ (and so $\nu \leq \lambda n/ |\R_n|$) the term on the right of this last equation is at least:
$$\left[ 1+ \left( 1- \frac{\lambda n}{|\R_n|} \right)^s \right] \left[ 1- \frac{\beta \lambda n}{|\R_n|}\right]^s$$
and since $s \leq |X_n|$, Inequality (\ref{firstfrac}) implies further  that:
\begin{equation}
\label{may1}
\frac{\theta_{(2f,h,s)}}{\theta_{(f,0,s)}}\geq   \left[ 1+ \left( 1- \frac{\lambda n}{|\R_n|} \right)^{|X_n|} \right] \left[ 1- \frac{\beta \lambda n}{|\R_n|}\right]^{|X_n|}.
\end{equation}
We now apply two inequalities:  $|X_n| \geq |\R_n|/n$ (from Eqn.~(\ref{asyeq})) and the generic inequality
$(1-t/k)^k \geq e^{-t}(1 - t/k)$, for $t>0$ and integer $k\geq 1$  (taking $k= |\R_n|/n$) to obtain:
$$\frac{\theta_{(2f,h,s)}}{\theta_{(f,0,s)}} \geq  (1+e^{-\lambda}(1 -\lambda n/|\R_n|))(e^{-\beta \lambda}(1 - \beta\lambda n/|\R_n|)).$$
Notice that the quantity on the right of this last equation is greater or equal to 1 when  $\beta$ satisfies Eqn.~(\ref{constraint}).
Thus,  $\theta_{(2f,h,s)} \geq \theta_{(f,0,s)}$, and this holds for all possible values of $s = |S(\R')|$.  Applying
this inequality in  Eqn.~(\ref{fir}) and subsequently in Eqn.~(\ref{mueq}) leads to the inequality:
$$\mu(2f, h) \geq \mu(f,0).$$
Since $\mu(f,0) = \mu(f)$, this completes the proof of Theorem~\ref{mainthm2}.
\end{proof}

{\bf Remarks}  
\begin{itemize}
\item[(i)]
As an application of Theorem~\ref{mainthm2}, we saw in Fig.~\ref{fig:compare} that for the uniform binary polymer model with $n=10$, the probability of a RAF 
is very close to 1  when $f=2$.  If we now put $\lambda = f/n =0.2$ into Eqn. (\ref{constraint}), this  gives an upper bound value for $\beta$ of $\beta\approx 3$.  In other words, for the uniform binary polymer model with $n=10$ we should expect at least as many $u$-RAFs when the catalysis rate is $2f=4$ and inhibition rate is at most $3f = 6$ as there are RAFs when $f=2$.  Indeed we will see in Section~\ref{sec:Inhibition}  that such $u$-RAFs do exist.  

An interesting feature of this example is that the inhibition rate is of similar magnitude (indeed larger) than the catalysis rate, and we expect to find $u$-RAFs, even though inhibition is in some sense a stronger condition than catalysis (i.e. if a reaction is inhibited by one molecule type then no matter how many molecule types catalyze that reaction it cannot be in any $u$-RAF).  However, it is easily seen that this requires some fine balancing in the choice of $h$ and $f$. Indeed, if  we assume that $h/f$ is greater or equal to some fixed value $c>0$, then $u$-RAFs can only arise when $f$ lies within  a certain window: If $f$  is too small there will be no RAFs (and so no $u$-RAFs) and when $f$ becomes too large the growth in $h (>cf)$ will eventually ensure that all RAFs fail to be $u$-RAFs.

\item[(ii)]
The approximation in the second half of Eqn. (\ref{constraint}) becomes quickly accurate as $n$ grows since  the additional terms in the first half  (involving $\lambda n/|\R_n|$ and $\beta\lambda n/|\R_n|$, where $\beta \leq n$) converge to zero at an exponential rate.

Notice also that Theorem~\ref{mainthm2} deals with the expected number of RAFs and $u$-RAFs, while Theorem~\ref{mainthm} is concerned with the probability of at least one RAF.  Calculating exact
non-asymptotic results for the latter quantity  (for RAFs or $u$-RAFs) appears to be a considerably harder exercise than calculating expected values.
\end{itemize}

\section{Levels of catalysis within a RAF}
\label{levels1}

In previous work, and above, we looked at the levels of catalysis in a CRS needed to obtain a given probability of finding RAF sets in random instances of the binary polymer model \cite{Hordijk:04,Mossel:05,Hordijk:10,Hordijk:11,Hordijk:14a}. However, these analyses considered only the overall level of catalysis required in the {\it full} CRS, expressed as the average number $f$ of reactions catalyzed per molecule type.  One question that has remained unexplored so far, is whether this overall level of catalysis is also reflected within the actual RAF sets that occur in these model instances, or whether those RAF sets, somehow, ``exploit'' local variations in the overall catalysis level. Here we investigate this question in more detail.

Using the standard notation, given a CRS $\Q=(X, R, C)$ and a RAF $\R' \subseteq \R$, let  $f_{\R}$ be the average number of reactions catalyzed per molecule type in the full CRS (i.e. the quantity we have so far denoted by $f$), and let $f_{\R'}$ be the average number of reactions in $\R'$ that are catalyzed by molecule types involved in $\R'$ (i.e. by molecule types that are produced by  reactions in $\R'$ or present in the food set $F$). We then generated 10,000 random instances of the binary polymer model for $n=10$, $t=2$, and various values of $p$ (for the uniform method) or $a$ (for the power law method) in such a way that the probability $P_n$ of finding RAF sets ranged roughly from 0.1 to 0.9. In these model instances we then measured and compared the levels of catalysis $f_{\R}$ in the full CRSs and $f_{\R'}$ in the RAFs that were found by our RAF algorithm within these model instances.

\subsection{Uniform catalysis}

Fig. \ref{fig:Unif_f} (left) shows the different levels of catalysis required for increasing probabilities $P_n$ of finding RAF sets for the uniform catalysis version of the model. The black line (``CRS'') shows the average observed $f_{\R}$ values for those model instances that {\it do} have a RAF set. The red line (``RAF'') shows the corresponding average observed $f_{\R'}$ values (i.e., the average level of catalysis {\it within} those RAFs). Surprisingly, the level of catalysis within RAFs is much lower than that in the full CRS of which they are a part, and in all cases this difference is statistically highly significant. Finally, the blue line (``noRAF'') shows the average observed $f_{\R}$ values for those model instances that do {\it not} have a RAF set. Not surprisingly, this level is also lower than that of model instances {\it with} a RAF. If the level of catalysis is too low, RAF sets will simply not exist. However, what is striking is that the level of catalysis {\it within} a RAF is even lower than that in model instances that do not have a RAF set at all. Again, this difference is statistically significant.

\begin{figure}[htb]
\centering
\includegraphics[width=13cm]{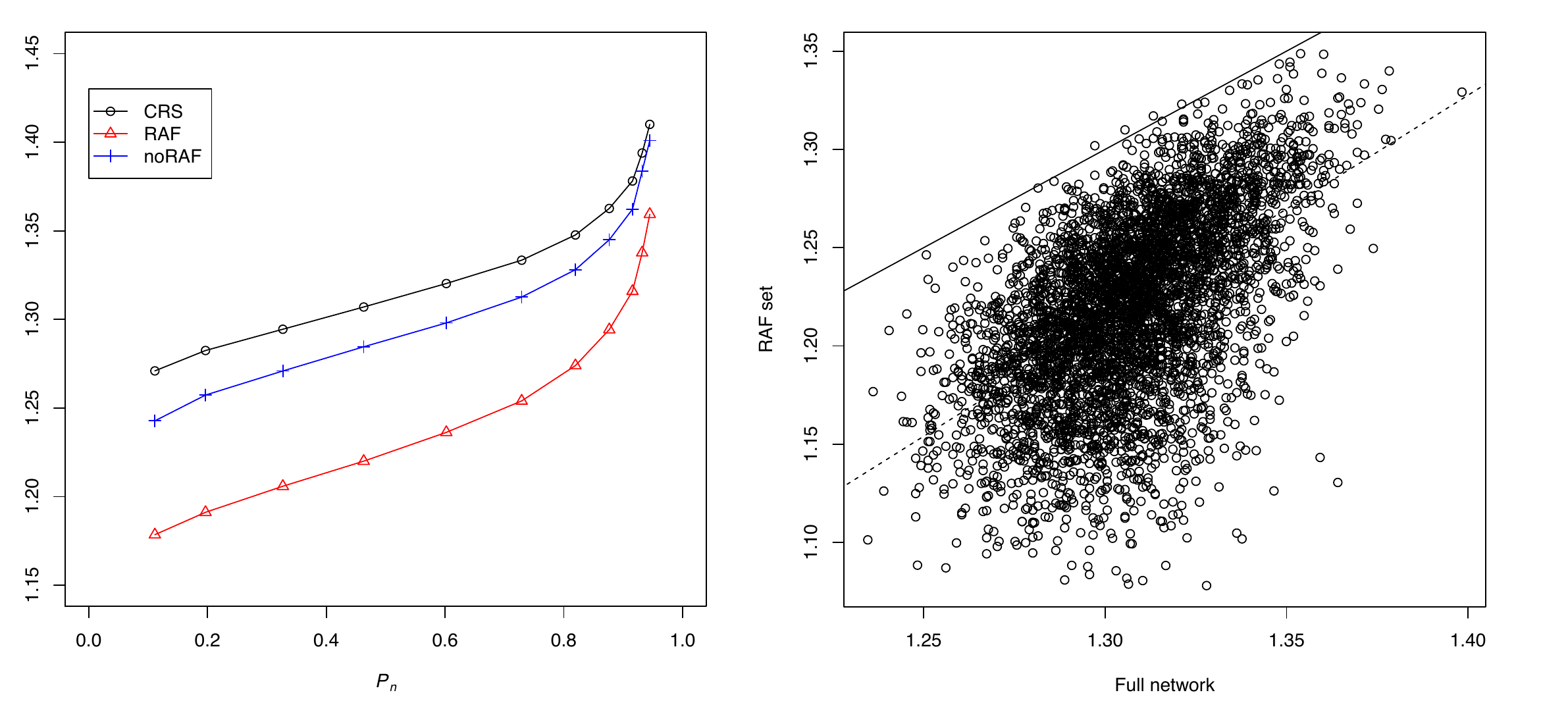}
\caption{{\em Left:}  The average $f_{\R}$ (``CRS'') and $f_{\R'}$ (``RAF'') for the uniform binary polymer model ($n=10$) over those model instances with a RAF set, and the average $f_{\R}$ (``noRAF'') over those instances without a RAF, plotted against the probability $P_n$ of finding RAF sets. {\em Right:} A scatter plot of $f_{\R'}$ vs. $f_{\R}$ for $P_n \approx 0.50$ for the uniform binary polymer model with $n=10$. The solid line represents $f_{\R'} = f_{\R}$, while the dashed line represents a linear fit.}
\label{fig:Unif_f}
\end{figure}

Initially, we had expected that if there was any difference between $f_{\R}$ and $f_{\R'}$, the catalysis level $f_{\R'}$ within the RAF set would be {\it higher} than the level $f_{\R}$ in the full CRS, assuming that the RAF set would exploit local variations in catalysis levels where they are higher than in other parts of the network. However, exactly the opposite appears to be the case. One possible explanation could be that if catalysts are assigned randomly to reactions, as in the binary polymer model, there will be many reactions that are catalyzed by more than one molecule type. To form a RAF set, in principle only one catalyst is necessary per reaction, so a RAF set might primarily include those molecule types and reactions that happen to have a more ``efficient'' (local) distribution of catalysis.

Fig. \ref{fig:Unif_f} (right) also shows a scatter plot of the $f_{\R'}$ values (``RAF set'') vs. the corresponding $f_{\R}$ values (``Full network'') for the $\sim$5,000 model instances that contained a RAF set  around $P_n=0.50$. The solid line indicates where $f_{\R'} = f_{\R}$. Apart from a handful of exceptions, the dots are all well below the line (i.e., $f_{\R} > f_{\R'}$ in almost all instances). The dashed line shows the linear fit from a standard regression analysis.

\subsection{Power law catalysis}

The same analysis was performed for the power law catalysis version of the binary polymer model. Fig. \ref{fig:Plaw_f} (left) shows the results for the average catalysis levels $f$. 
\begin{figure}[htb]
\centering
\includegraphics[width=13cm]{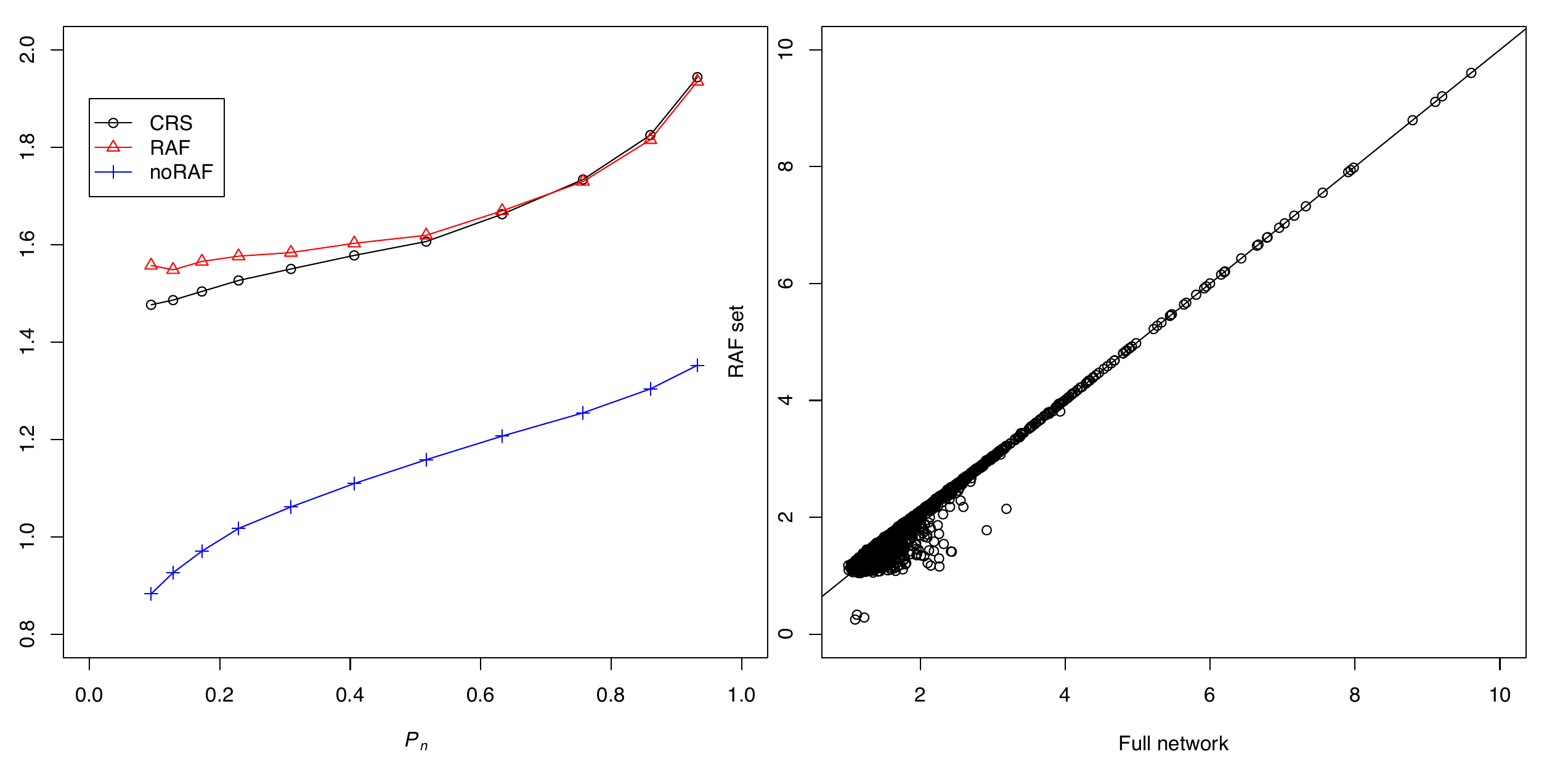}
\caption{{\em Left:} The average $f_{\R}$ (``CRS'') and $f_{\R'}$ (``RAF'') for the power-law binary polymer model ($n=10$) over those model instances with a RAF set, and the average $f_{\R}$ (``noRAF'') over those instances without a RAF, plotted against the probability $P_n$ of finding RAF sets. {\em Right:} A scatter plot of $f_{\R'}$ vs. $f_{\R}$ for $P_n \approx 0.50$ for the power-law binary polymer model with $n=10$. The solid line represents $f_{\R'} = f_{\R}$.}
\label{fig:Plaw_f}
\end{figure}
Here the results are partly in line with our initial expectation, in that RAF sets do indeed seem to ``exploit'' differences in the catalysis level where it is locally higher. With power law catalysis, there is always a chance that there are a few molecule types that catalyze many reactions (so-called ``hubs'', or highly connected nodes). Averaged over the entire network (with many molecule types that do not catalyze any reactions at all), this may not have a very large impact. However, a RAF set that is primarily constructed around these hubs will indeed have a higher level of catalysis than the full network of which they are part. This seems to happen up to a level of about $P_n=0.40$. After that, the differences between the full CRS and the RAF are not statistically significant anymore (at the 5\% level).

What is even more striking here is the difference in the average level of catalysis $f_{\R}$ between model instances that {\it do} (black line) and those that do {\it not} (blue line) contain a RAF set. Again, this is mostly due to assigning catalysis according to a power law, which results in a much larger variation in $f_{\R}$ values among model instances compared to the uniform catalysis method.

Also in Fig.~\ref{fig:Plaw_f} (right), the larger variance can be clearly seen in  the scatter plot of the $f_{\R'}$ values (``RAF set'') vs. the corresponding $f_{\R}$ values (``Full network'') for the $\sim$5,000 model instances that contained a RAF set around $P_n=0.50$. Note the much larger range of $f$ values compared to the uniform catalysis case above. The solid line again indicates where $f_{\R'} = f_{\R}$. Here, the majority of the dots are close to or slightly above this line, with only a large portion of dots below the line (i.e., $f_{\R} > f_{\R'}$) for smaller values of $f$.

\section{Growth rates in level of catalysis}
\label{levels2}

Two other issues that were only partially addressed in previous work are related to the growth rate in the level of catalysis $f$ needed for RAFs to arise with increasing maximum molecule length $n$ in the binary polymer model. First, although the mathematical results show that linear growth with $n$ is both necessary and sufficient, the constants involved in the 
upper and lower bounds are quite far apart.  Second, in the power law catalysis version of the model it seemed that there is no increase in the level of catalysis required at all (with increasing $n$), raising the question of whether this phenomenon persists for large values of $n$.  Both these issues are  resolved here.

\subsection{Theoretical vs. experimental rates}

There is a clear gap between the actual growth rate in level of catalysis needed for autocatalytic sets to form, as experimentally observed \cite{Hordijk:04} and the lower bounds predicted by theory (e.g. Eqns. (\ref{basiceq}) and (\ref{beq2})).    For the uniform catalysis version of the binary polymer model, both the theoretical and experimental rates were shown to follow a linear relationship in $n$ (the maximum molecule length) but the actual slopes differ. This is due to several simplifying assumptions in the theoretical proof.  We showed previously that at least one of those assumptions can partially explain this difference \cite{Hordijk:11}. Here, we show that a second simplifying assumption in the proof can explain most of the remaining difference.

The slope of the linear relationship as calculated by the lower bound on $P_n$ given by Theorem 4.1 (ii) in \cite{Mossel:05},  is given by $\ln\left(\frac{-1+\sqrt{17}}{16}\right) \approx 1.63$ (setting $P_n=0.50$ and $\kappa = t = 2$). However, the slope derived from the original experiments in \cite{Hordijk:04} is just under 0.02, a difference of two orders of magnitude. 
Notice, however, that the  bound in Theorem 4.1(i) in \cite{Mossel:05} is for  a RAF that contains all the molecule types in $X$, a much stronger requirement than is necessary in practice. Indeed, in \cite{Hordijk:11}, we showed that this assumption can explain a large part of this difference by performing experiments that require all molecule types in $X$ to be in the RAF. This provided a slope of 0.70 \cite{Hordijk:11}.

Also, the  bound in Theorem 4.1(i) requires only  {\it forward} (or ligation) reactions to be  considered (indicated by $\R_+$ in \cite{Mossel:05}). Recall that in the binary polymer model each ligation reaction has a corresponding cleavage (or {\it backward}) reaction, and if a molecule type is assigned as a catalyst to a forward reaction, it is also automatically assigned as a catalyst to the corresponding backward reaction. In other words, a forward reaction and the corresponding backward reaction always come in pairs, and are catalyzed by the same molecule type(s).  Backward (cleavage) reactions are useful though, as they can potentially provide shorter molecule types (bit strings) that can then be used to create necessary longer molecule types through one ligation reaction, whereas otherwise these longer molecule types may have to be built up through several consecutive ligation reactions, starting all the way from the food set.  

Indeed, it turns out that when only forward reactions are allowed in the binary polymer model, a higher level of catalysis $f$ is required to get the same probability of finding RAF sets compared to the model with both forward and backward reactions. Redoing the experiments using only forward reactions {\it and} requiring all molecule types in $X$ to be in the RAF, results in a slope of 1.51.

\begin{table}[htb]
\centering
\begin{tabular}{|l|l|}
\hline
Case         & Slope \\
\hline
standard     & 0.02 \\
all-$X$      & 0.70 \\
F \& all-$X$ & 1.51 \\
theory       & 1.63 \\
\hline
\end{tabular}
\caption{The slopes of the linear relationship for the growth rate in required level of catalysis, with increasing maximum molecule length $n$, as derived from the four different cases described in the text.}
\label{tab:slopes}
\end{table}

Table \ref{tab:slopes} summarizes these results, comparing the slopes for the four cases described: (i) the original experiments with the standard (uniform catalysis) binary polymer model (``standard''), (ii) the experiments requiring all molecule types in $X$ to be in the RAF (``all-$X$''), (iii) the experiments using only forward reactions {\it and} requiring all molecule types in $X$ to be in the RAF (``F \& all-$X$''), and (iv) the theoretical proof (``theory'').  Clearly, the additional assumption of only forward reactions can explain most of the remaining difference between the theoretical and the experimental slopes in the linear relationship for the growth in level of catalysis with increasing $n$.

\subsection{Power law catalysis}
In contrast to the uniform catalysis version of the binary polymer model, the power law catalysis version seemed to require {\it no} increase in the level of catalysis, with increasing $n$, to get a probability $P_n \approx 0.50$ of finding RAF sets.
\begin{figure}[htb]
\centering
\includegraphics[width=8cm]{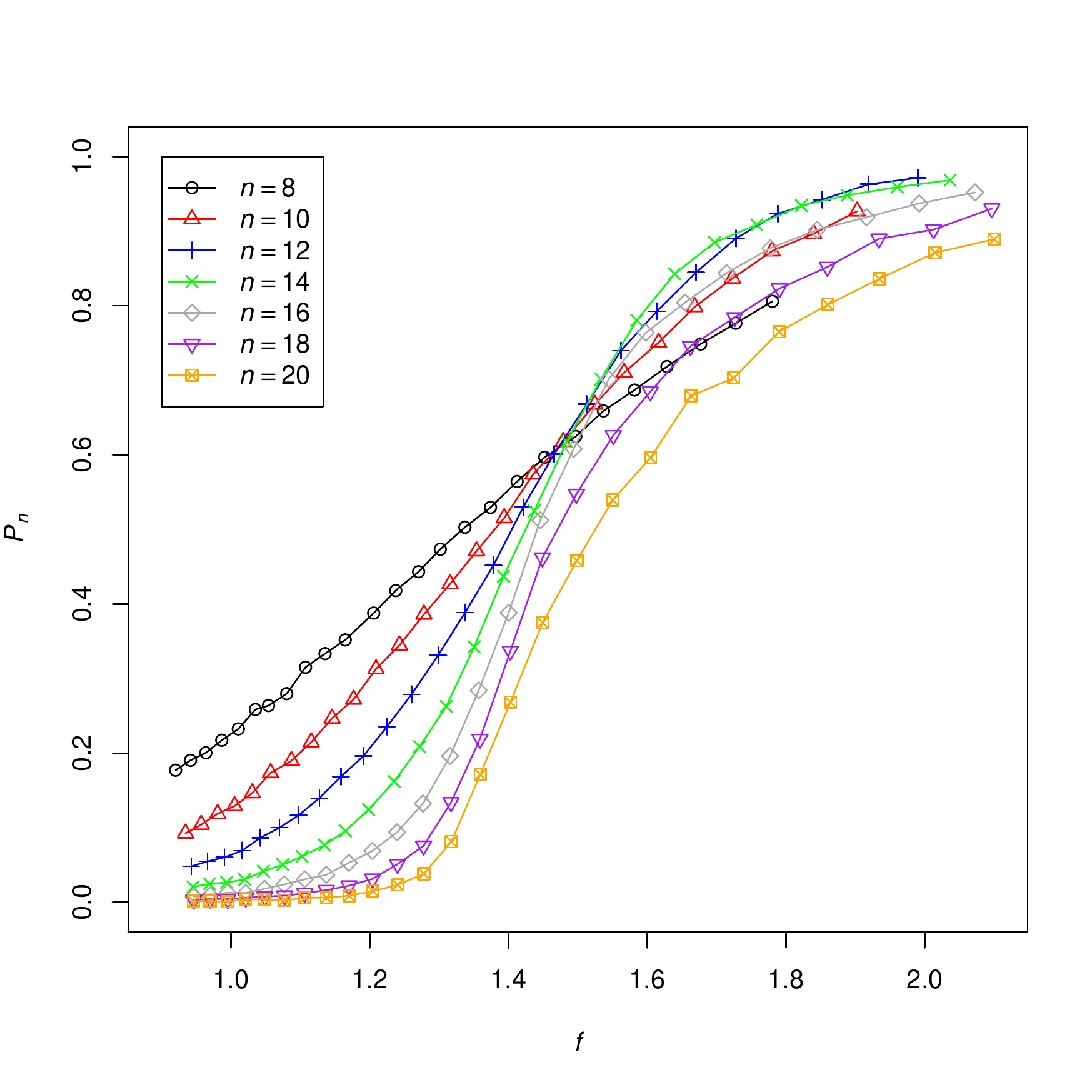}
\caption{The probability $P_n$ of finding RAF sets against the average level of catalysis $f$ for various values of the maximum molecule length $n$ in the binary polymer model with power law distributed catalysis.}
\label{fig:Plaw_P_f}
\end{figure}
 In Fig. 9 in \cite{Hordijk:14a}, $P_n$ vs. $f$ is plotted for various values of $n$ (up to $n=12$) for the power law case. However, rather than slowly moving to the right (i.e. towards higher $f$ values, as is the case for the uniform model \cite{Hordijk:04}),  the curves for increasing $n$ for the power law case cross each other around $f=1.5$ \cite{Hordijk:14a}. This result seems somewhat counter-intuitive, especially as Theorem \ref{mainthm}(ii) above implies that the curves must eventually start moving to the right as $n$ grows.
It turns out that this phenomenon is indeed a result of the relatively small values of $n$ (maximum molecule length) that were used in the initial experiments. We have redone the experiments with the power law catalysis version of the model, and also including larger values of $n$ (up to $n=20$). In this case, the curves stop crossing  and start to move slowly to the right for $n \geq 16$, as shown in Fig. \ref{fig:Plaw_P_f}.

\section{Inhibition} \label{sec:Inhibition}

In this final section, we return to the issue of inhibition, discussed earlier.  In biological systems, inhibitors play an important role in regulation (genetic, metabolic, or otherwise). As described in Section \ref{sec:BinPolModel} above, the notion of inhibition can be included in the RAF framework as well. However, in \cite{Mossel:05} it was shown that the problem of finding a $u$-RAF (or even determining whether or not one exists in a CRS allowing inhibition) is NP-hard. So, unlike the case with ``regular'' RAFs, we cannot expect to have an efficient algorithm for detecting $u$-RAFs in general CRSs that allow inhibition.

In \cite{Hordijk:15b}, we introduced a fixed-parameter tractable algorithm for finding $u$-RAFs (including a maximal $u$-RAF), with  the total number $m$ of inhibitors in the system as a parameter. If $m$ is small enough, the problem can still be solved efficiently, as we showed with an application to the binary polymer model and using $m=10$ \cite{Hordijk:15b}.

Here, we introduce an efficient iteration algorithm for finding $u$-RAFs, independent of the value of $m$. This algorithm is polynomial in the size $|\R|$ of the CRS, and provides a ``sufficient condition'' for a $u$-RAF to exist. In other words, if the iteration algorithm returns a non-empty result, the CRS contains a $u$-RAF, and therefore also a maximal $u$-RAF (although this could be larger than the $u$-RAF found by the algorithm). However, if the iteration algorithm returns an empty set, the CRS may still contain one or more $u$-RAFs (i.e., more work needs to be done to find out).

The iteration algorithm works as follows. Given a ``regular'' CRS plus an inhibition set $I$ (i.e., a set of molecule type--reaction pairs $(x,r)$ such that molecule type $x$ inhibits reaction $r$), perform the following steps:
\begin{enumerate}
  \item Let $\R_1$ be the ``regular'' RAF for the given CRS (i.e., not taking inhibition into account and using the standard RAF algorithm).
  \item Let $\R_2$ be the subset of $\R_1$ consisting of all reactions in $\R_1$ that are {\it not} inhibited by any of the molecule types produced by reactions in $\R_1$ or present in the food set $F$.
  \item Let $\R_3$ be the ``regular'' RAF contained within $\R_2$ (if any).
\end{enumerate}
Clearly, if $\R_3$ is non-empty, then it is a $u$-RAF for the given CRS (although it may not be a maximal $u$-RAF).

As a test of this algorithm, we searched for $u$-RAFs in the binary polymer model ($n=10$) for the setting  discussed in Remark (i) (following the proof of Theorem~\ref{mainthm2}).
\begin{table}[htb]
\centering
\begin{tabular}{|rr|}
\hline
$(1)$ & $(2)$ \\
\hline
6393 &  2547 \\
6415 &  2586 \\
6447 &  2729 \\
6410 &  2716 \\
6351 &  2581 \\
6359 &  2667 \\
6395 &  2718 \\
6397 &  2512 \\
6418 &  2631 \\
6438 &  2625 \\
\hline
\end{tabular}
\caption{The sizes of (1) the maximal RAF not taking inhibitors into account, and  (2) the $u$-RAF from the iterative algorithm, across 10 random instances of the uniform binary polymer model containing a RAF (with $n=10$, catalysis rate $f=4$ and inhibition rate $h=6$).}
\label{tab:m_all}
\end{table}
This involves a catalysis rate of $4$ and an inhibition rate of $6$. Theorem~\ref{mainthm2} coupled with Fig.~\ref{fig:compare} suggested that $u$-RAFs would likely exist, and 
applying the iteration algorithm over 10 random instances succeeded in detecting $u$-RAFs in all cases, as shown in Table~\ref{tab:m_all}.  The  fixed-pararameter
tractable $u$-RAF algorithm would be quite infeasible here, as $m=|X_{10}| = 2046$.

\subsection{Efficiency of iteration algorithm}

Since the ``standard'' RAF algorithm has a running time that is polynomial in the size $|\R|$ of the given CRS, the iteration algorithm above for finding $u$-RAFs is clearly also polynomial-time. We now compare the results of the iteration algorithm with that of the earlier exact algorithm on model instances in which the number $m$ of inhibitor molecule types is bounded (this  exact algorithm, by contrast, has a complexity that
grows exponentially with $m$). We then apply the iteration algorithm to model instances that are beyond the feasible reach of the exact method.

Note that in this section we continue to work in the uniform model, but now there are $m$ molecule types that are inhibitors, and for each of these the probability the molecule type inhibits any particular reaction is $q$. This is therefore quite different to the set-up in Section~\ref{impact} (and discussed above) where each molecule type has a constant probability of inhibiting each reaction. 

Using the same parameter values as in the previous study (i.e., $n=10$, $t=2$, $p=0.0000792$ [giving a probability $P_n$ of 0.50 to get ``regular'' RAFs]), $m=10$, and inhibition probabilities $q=10 \times p$ and $q=100 \times p$ \cite{Hordijk:15b}, the results are shown in Table \ref{tab:q10} (including the additional parameter value $q = 50 \times p$).

For each of the three cases (different values for the inhibition probability $q$), 10 instances of the model that contain a ``regular'' RAF were taken and both the exact and the iteration algorithm for finding $u$-RAFs were applied. In each table, the column labeled (1) shows the size of the (regular) maximal RAFs for each of the 10 instances. Column (2) shows the size of the maximal $u$-RAF that was found by the exact algorithm. Column (3) shows the size of the $u$-RAF found by the iteration algorithm.

\begin{table}[htb]
\centering
\begin{tabular}{|rrr|rrr|rrr|}
\hline
\multicolumn{3}{|c|}{$q = 10 \times p$} & \multicolumn{3}{|c|}{$q = 50 \times p$} & \multicolumn{3}{|c|}{$q = 100 \times p$} \\
(1) & (2) & (3) & (1) & (2) & (3) & (1) & (2) & (3) \\
\hline
1544 & 1510 & 1506 & 1339 & 1182 & 0 &1424 & 1290 & 1 \\
1459 & 1451 & 1451 & 1395 &    0 & 0 & 1645 & 1608 & 0  \\
1238 & 1206 &    0 & 1196 & 1181 &  0  & 1526 & 1473 & 0 \\
1732 & 1728 & 1728 & 1507 & 1469 & 1198 & 1447 & 0 & 0  \\
1447 & 1423 & 1423 & 1481 & 1461 & 1347 & 1441 & 1277 & 1 \\
1466 & 1457 & 1457 & 1480 & 1234 & 1150 & 1606 & 1583 & 0  \\
1478 & 1443 & 1306 & 1501 & 1469 &    0 &  1419 & 1397 & 0 \\
1516 & 1508 & 1508 & 1596 & 1509 & 1442 & 1457 & 1348 & 0 \\
1302 & 1293 & 1272 & 1319 & 1290 &  0 & 1147 & 1117 & 0 \\
1490 & 1444 & 1392 & 1642 & 1585 & 1350 & 1184 &   0 & 0  \\
\hline
\end{tabular}
\caption{For $m=10$, the sizes of (1) the maximal RAF not taking inhibitors into account, (2) the maximal $u$-RAF, and (3) the $u$-RAF from the iteration algorithm, for $q=10 \times p$ (left) and $q = 50 \times p$  (middle) and $q=100\times p$ (right), across ten random instances containing a RAF.}
\label{tab:q10}
\end{table}

For the case where $q=10 \times p$ (Table \ref{tab:q10}, left),  the iteration algorithm clearly performs almost as well as the exact algorithm. In all but one case, the iteration algorithm also finds a $u$-RAF, and its size is close to or often even equal to that of the maximal $u$-RAF found by the exact algorithm.

However, for $q=50 \times p$ (Table \ref{tab:q10}, middle), the situation is already quite different, with the iteration algorithm only finding a $u$-RAF in about half of the cases, and with a size significantly smaller than that of the maximal $u$-RAF found by the exact algorithm. Note that in four of the five cases where the iteration algorithm did {\it not} find a $u$-RAF, there actually was one.

Finally, for $q=100 \times p$ (Table \ref{tab:q10}, right) the situation is much worse still. In most cases the iteration algorithm does not find a $u$-RAF at all (or otherwise only a ``trivial'' one of size 1), while in 8 out of the 10 cases there actually is one.

Clearly, the performance of the iteration algorithm depends on the ratio of the rates of catalysis to inhibition.  Recall that the catalysis rate $f$ is defined as the average number of reactions catalyzed per molecule type, which, in the binary polymer model, can be expressed as \cite{Hordijk:04}:
\[ f = \frac{p|X||\R|}{|X|} = p|\R|. \]
Similarly, the inhibition rate is  the average number of reactions inhibited per molecule type:
\[ h = \frac{qm|\R|}{|X|} = \frac{spm|\R|}{|X|}, \]
where $s$ is the factor by which the catalysis probability $p$ is multiplied to get the inhibition probability $q$ (i.e., $s=10,50,100$ in our experiments). Therefore, the ratio $f/h$ of the level of catalysis to the level of inhibition is:
\[ \frac{f}{h} = \frac{|X|}{sm}. \]

For $n=10$, the number of molecule types in the CRS is $|X|=2046$. For $s=10$ and $m=10$, this ratio $f/h$ is about one order of magnitude in size, and the iteration algorithm still works quite well. But for $s=100$ and $m=10$, this ratio is close to one (i.e., the two levels are similar), and the iteration algorithm may fail to detect $u$-RAFs when they exist, or find ones that are much smaller than the maximal ones.

However, the main advantage of the iteration algorithm is that it can be used for any value of $m$ (number of inhibitors). For example, for $m=100$, the exact algorithm is not feasible anymore, but the iteration algorithm can still be used as before (its running time does not depend on $m$). Using the same parameter values, but with $m=100$ and $q=p$ (giving a ratio $f/h$ about one order of magnitude in size), the results are as shown in Table \ref{tab:q100}.  In all cases, the iteration algorithm found a $u$-RAF, although it is not known whether its size is close to the maximal $u$-RAF.

\begin{table}[htb]
\centering
\begin{tabular}{|rrr|rrr|}
\hline
\multicolumn{3}{|c|}{$m=100$, $q = p$} \\
(1) & (2) & (3) \\
\hline
1417 & ? & 1402 \\
1278 & ? & 1161 \\
1225 & ? & 1211 \\
1445 & ? & 1330 \\
1347 & ? & 1335 \\
1451 & ? & 1389 \\
1344 & ? & 1285 \\
1531 & ? & 1492 \\
1508 & ? & 1451 \\
1312 & ? & 1175 \\
\hline
\end{tabular}
\caption{The sizes of (1) the maximal RAF without taking inhibitors into account, (2) the maximal $u$-RAF, and (3) the $u$-RAF from the iteration algorithm, for  $m=100$, $q = p$ (right), for 10 random instances containing a RAF.}
\label{tab:q100}
\end{table}

In conclusion, the performance of the iteration algorithm depends largely on the ratio $f/h$ of the level of catalysis relative to the level of inhibition. If the ratio is large enough (one order of magnitude or larger), the iteration algorithm performs quite well. If this ratio is too small (close to equality), the iteration algorithm may not detect $u$-RAFs when they exist.
However, the advantage of the iteration algorithm is that it can be used for any value of $m$ (i.e., for any number of inhibitors) to determine whether there exist uninhibited RAFs in CRSs that also allow inhibition.

\section{Future work}

Our results suggest a number of plausible conjectures and questions for further study.
Firstly, the sharpness of the transition from $P_n=0$ to $P_n=1$ for the uniform model (and perhaps also the power law model), evident in Fig.~\ref{fig:compare}, suggests that a 0-1 law (analogous to Eqn.(\ref{allR})) may hold. Namely,  for some constant $\gamma$ under the uniform model, we have:
\begin{equation}
\label{allR2}
\lim_{n \rightarrow \infty} \PP(\Q_n \mbox{  has a RAF}) =
 \begin{cases}
1, &  \mbox{ if } f \geq \lambda n \mbox{ for } \lambda > \gamma;\\
0, & \mbox{ if } f \leq \lambda n \mbox{ for } \lambda < \gamma.
\end{cases}
\end{equation}
Settling this conjecture seems a reasonable objective.  Clearly, such  a sharp transition does not hold for all distribution-based models, since it Eqn.(\ref{pneq}) shows that it fails for the all-or-nothing model and it might possibly also fail for the sparse model. Simulations from \cite{Hordijk:04} suggest that $\gamma$ might be close to 0.02.

The sizes of the smallest RAFs at catalysis levels where RAFs first arise is also an interesting question. Here, the uniform model behaves quite differently from the sparse and all-or-nothing models: small (sub-exponential size) RAFs in the uniform model are highly unlikely when they first emerge, while in the other two models they will always be present.  The situation for the power law model is less clear. Simulations and the application of the algorithm to find irreducible RAFs suggest that there are smaller RAFs in this model than in the uniform \cite{Hordijk:14a}, but a mathematical analysis of the existence (or not) of small RAFs in the power law setting could be of interest.

Finally, although we showed that the sparse model is likely to have RAFs of, at most, quadratic size in $n$ (when RAFs first appear), it is not clear if this can be improved to order $n$. In particular, is it likely that there are RAFs in which a single molecule type catalyzes some sequence of reactions that suffice to construct that molecule? A simple mathematical argument shows that this becomes increasingly unlikely as $n$ grows if we restrict ourselves to just ligation reactions (results not shown). However, the situation is less clear without this restriction.

\section*{Acknowledgments}

WH thanks the Konrad Lorenz Institute Klosterneuburg for financial support in the form of a fellowship.
We also thank the NLHPC, Santiago, Chile (www.nlhpc.cl) for allowing us to use their high performance computing infrastructure.

\bibliographystyle{abbrv}
\bibliography{CatLevel}

\section*{Appendix: Proof of Eqn.~(\ref{allR})}
It suffices to establish the following stronger result: if $\lambda = \ln \kappa + c>0$, and $f = \lambda n^2$, then for $\beta :=\frac{ \kappa}{\kappa-1}$, we have:
\begin{equation}
\label{allraf} 
\PP(\R_n \mbox{ is a RAF in } \Q_n) \sim \exp (-(n\beta - \beta^2)e^{-cn}),
\end{equation}
which converges to 0 or 1 depending on whether or not $c<0$ or $c>0$ respectively.
To establish Eqn.~(\ref{allraf}),  note that for $f = ( \ln \kappa + c) n^2$,  Eqn.~(\ref{qeq}) gives:
$q_{*} \sim \kappa^{-n}e^{-cn}.$
Now,  the entire reaction set $\R_n$ is a RAF precisely if every reaction from $\R_n$ is catalyzed by at least one molecule type, and so
$$\PP(\R_n \mbox{ is a RAF in } \Q_n) \sim \left[1- \frac{1}{\kappa^n} e^{-cn}\right]^{|\R_n|}  \sim \exp\left(-\frac{|\R_n|}{\kappa^n} e^{-cn}\right).$$
Eqn.~(\ref{allraf}) follows by noting that $|\R_n|/\kappa^n \sim n\beta -\beta^2$ from  Eqn.~(3) of \cite{Mossel:05}.
\hfill$\Box$

\end{document}